\documentclass[runningheads]{llncs}

\usepackage{graphicx} 
\usepackage{amsmath, amssymb}

\let\doendproof\endproof
\renewcommand\endproof{~\hfill$\qed$\doendproof}

\newcommand{\opt}[2]{\ensuremath{\mathsf{OPT}_{#1}(#2)}}
\newcommand{\D}{\ensuremath{P}}
\newcommand{\V}{\ensuremath{V}}
\newcommand{\diam}[1]{diam({#1})}
\newcommand{\VD}[1]{\text{VD}({#1})}
\newcommand{\DT}[1]{\text{DT}({#1})}

\newcommand{\PB}{\textrm{PB}}
\newcommand{\optL}{\textsc{Lnr}}
\newcommand{\optS}{\textsc{Slb}}

\newcommand{\UpGamma}{\mathrm{\Gamma}}
\newcommand{\UpPi}{\mathrm{\Pi}}

\begin{document}

\title{The Polygon Burning Problem}

\author{William Evans \and Rebecca Lin}

\authorrunning{W. Evans and R. Lin}

\institute{
University of British Columbia, Vancouver, Canada\\
\email{will@cs.ubc.ca, ryelin@student.ubc.ca}}

\maketitle          


\begin{abstract}

Motivated by the $k$-center problem in location analysis, we consider the \emph{polygon burning} (PB) problem: Given a polygonal domain $\D$ with $h$ holes and $n$ vertices, find a set $S$ of $k$ vertices of $\D$ that minimizes the maximum geodesic distance from any point in $\D$ to its nearest vertex in $S$. Alternatively, viewing each vertex in $S$ as a site to start a fire, the goal is to select $S$ such that fires burning simultaneously and uniformly from $S$, restricted to $\D$, consume $\D$ entirely as quickly as possible. We prove that PB is NP-hard when $k$ is arbitrary. We show that the discrete $k$-center of the vertices of $\D$ under the geodesic metric on $\D$ provides a $2$-approximation for PB, resulting in an $O(n^2 \log n + hkn \log n)$-time $3$-approximation algorithm for PB. Lastly, we define and characterize a new type of polygon, the sliceable polygon. A sliceable polygon is a convex polygon that contains no Voronoi vertex from the Voronoi diagram of its vertices. We give a dynamic programming algorithm to solve PB exactly on a sliceable polygon in $O(kn^2)$ time. 

\keywords{$k$-center \and Polygon covering \and Voronoi diagram}
\end{abstract}


\section{Introduction}
\label{section:introduction}

Given a set $S$ of $n$ points representing clients or demands, the \emph{$k$-center} problem asks to determine a collection $C$ of $k$ center points for placing facilities so as to minimize the maximum distance from any demand to its nearest facility. Geometrically speaking, the goal is to find the centers of $k$ equal-radius balls whose union covers $S$ and whose common radius, the \emph{radius} of the $k$-center, is as small as possible. This paper assumes the discrete version of the $k$-center problem where centers are selected from $S$. 

The $k$-center problem is NP-hard when $k$ is an arbitrary input parameter and NP-hard to approximate within a factor of  $2-\epsilon$ for any $\epsilon > 0$. However, there exist several $2$-approximation algorithms that hold in any metric space~\cite{Gonzalez1985,HochbaumShmoys1985}. Gonzalez, for one, gave a greedy approach: Select the first center from $S$ arbitrarily, and while $\vert C \vert < k$, repeatedly find the point in $S$ whose minimum distance to the chosen centers is maximized and add it to $C$. 

In many real-world applications, demands are not restricted to a discrete set but may be distributed throughout an area. Consider, for example, installing charging stations in a warehouse so that the worst-case travel time of robots to their nearest stations is minimal. In practice, regions of demand are often modelled using polygonal domains. A \emph{polygonal domain} $\D$ with $h$ holes and $n$ vertices is a connected region whose boundary $\partial \D$ comprises $n$ line segments that form $h+1$ simple closed polygonal chains. If $\D$ is without holes, then it is a \emph{simple polygon}. We define the \emph{geodesic distance} $d(s, t)$ between any two points $s, t \in \D$ to be the Euclidean length of the shortest path connecting $s$ and $t$ that is contained in $\D$. 

Given a polygonal domain $\D$, the geodesic $k$-center problem on $\D$ asks to find a set $C$ of $k$ points in $\D$ that minimizes the maximum geodesic distance from any point in $\D$ to its closest point in $C$. We call $C$ the \emph{$k$-center} of $\D$. Asano and Toussaint~\cite{AsanoToussaint1985} gave the first algorithm for computing the $1$-center of a simple polygon with $n$ vertices; it runs in $O(n^4 \log n)$ time. This result was later improved by Pollack et al.~\cite{Pollack1989} to $O(n \log n)$, and recently, Ahn et al.~\cite{Ahn2016} presented an optimal linear-time algorithm. Following these explorations, Oh et al. ~\cite{OhDeCarufelAhn2018} gave an $O(n^2 \log^2 n)$-time algorithm for computing the $2$-center of a simple polygon. However, it appears that no results are known for $k > 2$ in the case of simple polygons. Likewise, for polygons with one or more holes, results are limited: only the $1$-center problem has been solved with a running time of $O(n^{11} \log n)$~\cite{Wang2016}.

In practice, facilities are often restricted to feasible locations. Hence, there has been some interest in constrained versions of the geodesic $k$-center problem on polygonal domains. Oh et al.~\cite{OhSonAhn2016} considered the problem of computing the $1$-center of a simple polygon constrained to a set of line segments or simple polygonal regions in the polygon. Du and Xu~\cite{DuXu2014}  proposed a 1.8841-approximation algorithm for computing the $k$-center of a convex polygon $P$ with centers restricted to the boundary of $P$. 

In this paper, we consider a new variant of the geodesic $k$-center problem that restricts facilities to the vertices of the given polygonal domain. Unlike the original problem and the constrained versions above, our problem is a combinatorial optimization problem: We draw centers from a finite set of points rather than a region in the plane. Viewing each vertex as a potential site to start a fire, we arrive at the following problem formulation we adopt in this paper. 

\begin{definition}[Polygon Burning]
\label{definition:pb}
Given a polygonal domain $\D$ with $h$ holes and $n$ vertices and an integer $k \in [1, n]$, find a set $S$ of $k$ vertices of $\D$ such that $\D$ is consumed as quickly as possible when burned simultaneously and uniformly from $S$.
\end{definition}

Section~\ref{section:preliminaries} is devoted to the background required for our study. In Section~\ref{section:hardness}, we prove that PB is NP-hard when $k$ is part of the input. In Section~\ref{section:approximation}, we show that the $k$-center of the vertices of $\D$ under the geodesic metric on $\D$ provides a $2$-approximation for PB on $\D$. This result leads to an $O(n^2 \log n + hkn \log n)$-time $3$-approximation algorithm for PB. Finally, given the NP-hardness of PB in general, we shift our focus to restricted instances. In Section~\ref{section:sliceable}, we consider convex polygons that contain no Voronoi vertex from the Voronoi diagram of their vertices. We call such instances sliceable. Their structure admits a natural ordering of separable subproblems, permitting an exact $O(kn^2)$ algorithm using the dynamic programming technique. 


\section{Preliminaries} 
\label{section:preliminaries}

Unless stated otherwise,
the distance metric $d$ we use on a polygonal domain $\D$ is the geodesic metric on $\D$.
The \emph{diameter} of $\D$, $\diam{\D}$, is the largest distance between any two points in $\D$.

Let $S = \{s_1, s_2, \dots, s_k\}$ be a set of $k$ points, called \emph{sites} or \emph{burn sites}, in a region $R$. The Voronoi diagram $\text{VD}_R(S)$ of $S$ is the subdivision of $R$ into $k$ Voronoi regions, one per site $s_i \in S$, such that any point in the Voronoi region of $s_i$ is closer to $s_i$ (using the geodesic metric on $R$) than to any other site in $S$. We refer to $\text{VD}_{\mathbb{R}^2}(S)$ as $\text{VD}(S)$. 

Consider a polygonal domain $\D$ with vertices $\V = \{v_1, v_2, \dots, v_n\}$. Let $S \subseteq \V$ be a selection of $k$ burn sites. Each Voronoi region $\D_i$ in the Voronoi diagram $\text{VD}_{\D}(S)$ of $S$ is the set of points in $\D$ burned by the fire from site $s_i \in S$. We associate with each point $p$ in $\D_i$ the time it burns, which is the distance travelled by the fire from $s_i$ to $p$. It follows that $\D$ burns in time $t_{S}(\D) = \max_{s_i \in S} \max_{p \in \D_i} d(s_i, p)$. As described in Definition~\ref{definition:pb}, PB asks to find a set $S \subseteq V$, $\vert S \vert = k$, that minimizes $t_{S}(\D)$.
We let $S_k(\D)$ denote such an optimizing set and 
let $\opt{k}{\D}$ be the minimum burning time of $\D$. 

A \emph{geodesic disk} of radius $r$ centered at a point $p \in \D$ is the set of points in $\D$ at most geodesic distance $r$ from $p$. By definition, the union of $k$ geodesic disks of radius $\opt{k}{\D}$ centered at the sites in $S_k(\D)$ contains $\D$.
Observe that $\diam{\D} \leq 2k \cdot \opt{k}{\D}$ since $\D$ cannot be covered by $k$ geodesic disks of radius $\opt{k}{\D}$ otherwise. 
The time to burn $\D$ given any non-empty selection of burn sites is at most $\diam{\D}$.
Hence any non-empty selection of burn sites in $\V$ gives a $2k$-approximation for $\PB$ with $k$ sites on $\D$. 


\section{Hardness} 
\label{section:hardness} 

In this section, we show that PB is NP-hard on polygonal domains. 
We reduce from 4-Planar Vertex Cover (4VPC): Given a planar graph $G$ with max-degree four and an integer $\kappa$, does $G$ contain a vertex cover (i.e., a set of vertices $C \subseteq V(G)$ such that every edge in $G$ contains at least one vertex in $C$) of size at most $\kappa$? This problem is known to be NP-hard~\cite{GareyJohnson1979}. 

Given an instance $G,\kappa$ of 4PVC, we construct an equivalent instance of PB. First we compute an orthogonal drawing $\UpGamma$ of $G$ with $O(n)$ bends on an integer grid of $O(n^2)$ area (Figure~\ref{figure:graph-drawing}a) using an $O(n)$-time algorithm due to Tomassia and Tollis~\cite{TamassiaTollis1989}. Every edge $uv \in E(G)$ is represented as a sequence of connected line segments $\overline{p_1 p_2}, \overline{p_2 p_3}, \dots, \overline{p_{i-1} p_i}$ in $\UpGamma$, denoted $\UpGamma(uv)$, where $p_1 = \UpGamma(u)$ and $p_i = \UpGamma(v)$ correspond to the endpoints of $uv$ and $p_2, \dots, p_{i-1}$ are \emph{bends} in $\UpGamma(uv)$. The length $\vert \UpGamma(uv) \vert$ of $\UpGamma(uv)$ is the sum of the lengths of its line segments. 

\begin{figure}
\centering
\includegraphics[width=.975\textwidth]{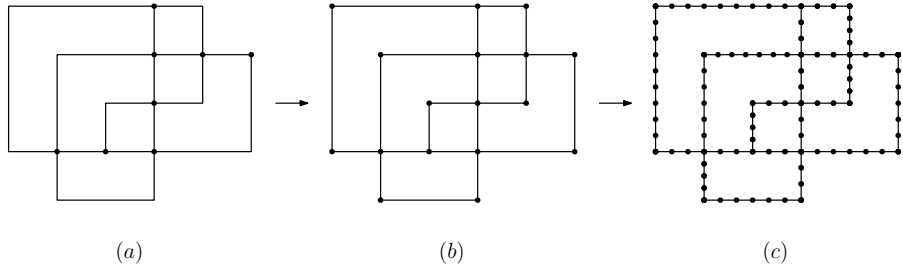}
\caption{(a) A planar orthogonal grid drawing $\UpGamma$ of $G$, (b) a straight-line grid drawing (step 1), and (c) the drawing $\UpPi$ of the subdivision $H$ of $G$ satisfying Property~\ref{property:edge-length-restriction} and \ref{property:vertex-cover-subdivsion} (step 2).}
\label{figure:graph-drawing}
\end{figure}

Next we transform $\UpGamma$ into a constrained straight-line drawing $\UpPi$ of a subdivision $H$ of $G$ in two steps. First we add a vertex at every bend in $\UpGamma$ (Figure~\ref{figure:graph-drawing}b). Then we replace each segment $\overline{p_j p_{j+1}}$ ($1 \leq j < i$) along $\UpGamma(uv)$ with either $3 \vert \overline{p_j p_{j+1}} \vert$ or $3 \vert  \overline{p_j p_{j+1}} \vert + 1$ equal-length edges depending on the parity required to ensure that the overall number $\ell_{uv}$ of segments along $\UpGamma(uv)$ is odd (Figure~\ref{figure:graph-drawing}c). Property~\ref{property:edge-length-restriction} and \ref{property:vertex-cover-subdivsion} follow from these steps. Property~\ref{property:vertex-cover-subdivsion} is due to the fact that a double subdivision of an edge in $G$ increases the size of any vertex cover of $G$ by one. 

\begin{property}
\label{property:edge-length-restriction}
For every $uv \in E(H)$, $\frac{1}{4} \leq \vert \UpPi(uv) \vert \leq \frac{1}{3}$.
\end{property}

\begin{property}
\label{property:vertex-cover-subdivsion}
$G$ has a vertex cover of size $\kappa$ if and only if $H$ has a vertex cover of size $K(G) := \kappa + \frac{1}{2} \sum_{uv \in E}(\ell_{uv} - 1)$. 
\end{property}

Finally, we convert $\UpPi$ into a polygonal domain $P(G)$ by thickening each line segment in $\UpPi$ as follows. For every vertex $v \in V(H)$, we replace $\UpPi(v)$ with a set $S(v)$ of four vertices at $\UpPi(v) + (-\epsilon, \epsilon)$, $\UpPi(v) + (\epsilon, \epsilon)$, $\UpPi(v) + (\epsilon, -\epsilon)$, and $\UpPi(v) + (-\epsilon, -\epsilon)$, where $\epsilon < \frac{1}{120}$ is a fixed constant. Let $R(uv)$ denote the convex hull of $S(u) \cup S(v)$. We define $P(G)$ to be the union of the collection of regions $R(uv)$ for all $uv \in E(H)$. 

It is straightforward to verify that the above transformation of an instance $G$ of 4PVC to an instance $P(G)$ of PB runs in $O(n)$ time. Furthermore, $P(G)$ has $O(n)$ vertices, and the number of bits required in the binary representation of each vertex coordinate is bounded by a polynomial in $n$. It remains to demonstrate that:

\begin{lemma}
$G$ has a vertex cover of size at most $\kappa$ if and only if $P(G)$ can be burned in time $\frac{1}{3} + 3 \epsilon$ using $K(G)$ sites. 
\end{lemma}

\begin{figure}
\centering
\includegraphics[width=.79\textwidth]{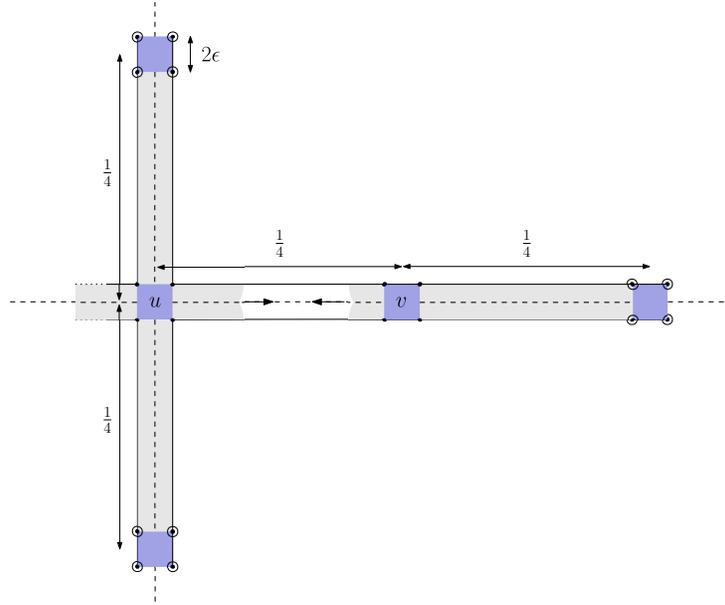}
\caption{A scenario where $R(uv)$ is burned the quickest assuming that no sites (circled) are selected from either $S(u)$ or $S(v)$. The two dashed lines are the only integer grid lines in the figure.
} 
\label{figure:hardness-lemma}
\end{figure}

\begin{proof}
It suffices to show that for any $uv \in E(H)$, $R(uv)$ can be burned in time $\frac{1}{3} + 3 \epsilon$ if and only if at least one vertex in $S(u) \cup S(v)$ is a burn site. The forward direction follows from observing that $\frac{1}{3} + 3 \epsilon$ is a loose upper bound on the burning time of $R(uv)$ given that a site is located in either $S(u)$ or $S(v)$ (Property~\ref{property:edge-length-restriction}). For the reverse direction, suppose no vertices in $S(u)$ or $S(v)$ are selected. We obtain a lower bound on the burning time of $R(uv)$ by considering the scenario where $R(uv)$ is burned the quickest: First, for each vertex $w \in H$ adjacent to either $u$ or $v$, let every vertex in $S(w)$ be a burn site. Second, assume $u$ and $v$ have as many adjacent edges as possible in $E(H)$ to assist in burning $R(uv)$. At most one of these two adjacent vertices can have degree greater than two since at most one is on the integer grid, and this vertex, say $u$, can have degree at most four. The other vertex $v$ can have degree two, but its adjacent edges must be colinear in the drawing.
Finally, suppose all these edges are as short as possible in the drawing $\UpPi$ ($\frac{1}{4}$ by Property~\ref{property:edge-length-restriction}). We find that the burning time of $R(uv)$, if no vertex in $S(u)$ or $S(v)$ is a site, is bounded below by $\frac{3}{8} - 2 \epsilon > \frac{1}{3} + 3 \epsilon$ (see Figure~\ref{figure:hardness-lemma}).
The lemma then follows from Property~\ref{property:vertex-cover-subdivsion}.
\end{proof}

As a result, we obtain: 

\begin{theorem}
\label{theorem:hardness} 
PB is NP-hard on polygonal domains. 
\end{theorem}


\section{Approximation by a $k$-Center}
\label{section:approximation}

We present a straightforward $3$-approximation algorithm for PB by considering the $k$-center problem described in the introduction. 

\begin{theorem}
\label{theorem:approx-by-k-center}
The radius of a $k$-center of the vertices $\V$ of $\D$, using the geodesic metric on $\D$, provides a $2$-approximation of $\opt{k}{\D}$.
\end{theorem}

\begin{proof}
Let $C \subseteq \V$ denote a $k$-center of $\V$ and let $r$ denote its radius. Observe two facts: First, $\opt{k}{\D} \geq r$ since
$\D \supseteq V$. Second, each point $p \in \D$ is within $\opt{k}{\D}$ of a vertex $v$ of $\D$, and $v$ is at most $r$ from some center $c$ in $C$. Therefore, by the triangle inequality, $d(p, c) \leq \opt{k}{\D} + r \leq 2 \opt{k}{\D}$ as desired. 
\end{proof}

\begin{corollary}
Applying Gonzalez's greedy $2$-approximation algorithm for finding a $k$-center of $\V$ yields an $O(n^2 \log n + hkn \log n)$-time $3$-approximation algorithm for PB on $\D$ that uses $O(n^2)$ space. 
\end{corollary}

\begin{proof}
The $2$-approximation algorithm provides an approximate $k$-center of $V$ whose radius $r'$ is at most $2r$ where $r$, as in the above proof, is the optimal $k$-center radius.
Following that proof, this yields a $3$-approximation.
The time and space complexity are due to performing $O(kn)$ geodesic distance queries on $\D$ using an algorithm by Guo et al.~\cite{Guo2008}.
Note,
if $\D$ is simple, then a $3$-approximation for PB can be found in $O(kn\log n)$ time using $O(n)$ space
by the faster geodesic distance queries of Guibas and Hershberger~\cite{GuibasHershberger1989}. 
\end{proof}


\section{Sliceable Polygons}
\label{section:sliceable} 

\begin{definition}
A \emph{sliceable} polygon $\D$ is convex and contains no Voronoi vertex from the Voronoi diagram $\VD{\V}$ of its vertices $\V$. 
\end{definition}

Every Voronoi edge in $\VD{V}$ that intersects $\D$ slices through $\D$ (Figure~\ref{figure:sliceable}). We can solve PB on $\D$ using dynamic programming, as $\D$ admits a total ordering of vertices with the property that if $u < v < w$ are burn sites, then the region of $\D$ burned by $u$ does not share a boundary with the region of $\D$ burned by $w$ (Lemma~\ref{lemma:ordering}). We start with a simple example that indicates the use of this property. 

\begin{figure}
\centering
\includegraphics[width=.5\textwidth]{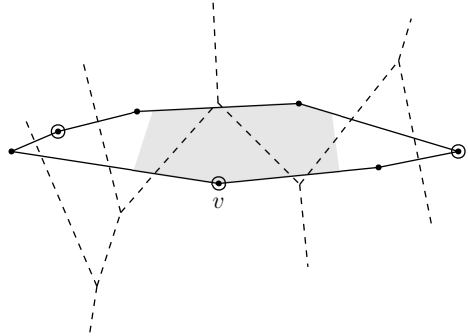}
\caption{A sliceable polygon $\D$ overlaid with the Voronoi diagram (dashed) of its vertices. By Lemma~\ref{lemma:ordering}, the region of $\D$ (shaded) burned by a site $v$ separates the regions burned by sites (circled) before $v$ in the ordering from regions burned by sites after $v$, no matter what those sites are. This holds for every $v$.} 
\label{figure:sliceable}
\end{figure}

\subsection{Polygons in One Dimension} 

Let $\D$ be a $1$-dimensional polygon with $n$ vertices $v_1, v_2, \dots, v_n$ ordered by x-coordinate. Let $P[i, j]$ be the segment of $P$ from $v_i$ to $v_j$. The minimum time to burn $\D$ using $k$ sites is 
\[
\opt{k}{\D}  =
\begin{cases}
\min_{i \in [n]} \max\{ d(v_1,v_i) , \optL(i,k-1)\} & \text{if $k>0$}\\
\infty & \text{otherwise,}
\end{cases} 
\]
where $d(v_1, v_i)$ is the time to burn $\D[1, i]$ from site $v_i$ and $\optL(i,k)$ denotes the minimum time to burn $\D[i,n]$ using $k$ sites in addition to $v_i$. If $k > 0$, then $\optL(i,k)$ is achieved by choosing the next site $v_j$ ($i<j\leq n$) to minimize the larger of two values: (i) the time $d(v_i,v_j)/2$ to burn $\D$ between $v_i$ and $v_j$ and (ii) the minimum time to burn $\D[j,n]$ knowing $v_j$ is a burn site with $k-1$ burn sites remaining. If $k = 0$, no sites are allowed beyond $v_i$, in which case the minimum time to burn $\D[i, n]$, with $v_i$ a burn site, is $d(v_i, v_n)$. 
\[
\optL(i,k) =
\begin{cases}
  \min_{i < j \leq n} \max\{ d(v_i,v_j)/2 , \optL(j,k-1) \} &
  \text{if $k > 0$,}\\
  d(v_i, v_n) & \text{otherwise.}
\end{cases}
\]
This recurrence relation relies only on the property that any burn site preceding the burn site $v_i$ is farther from every point in $\D[i,j]$ than $v_i$ for $j > i$. We will prove a similar property for sliceable polygons. 

A dynamic programming algorithm follows directly from the recurrence.

\begin{theorem}
PB can be solved in $O(kn^2)$ time on a $1$-dimension polygon with $n$ vertices. 
\end{theorem}

\begin{proof} 
(Sketch) Use dynamic programming.  Two observations hold on each iteration of the algorithm: (i) The choice of the following site $v_j$ is unaffected by the sites selected before the current site $v_i$, and (ii) we evaluate every possible choice $v_j$ and take the best amongst them. The natural ordering of subproblems implied by (i) combined with the virtue of an exhaustive search as noted in (ii) allows us to successfully compute the solution to the original problem from the solutions to the recursive subproblems. 

The algorithm populates a table of size $O(kn)$. To fill each entry, it computes the minimum of $O(n)$ previous entries.  Therefore, the total running time is $O(kn^2)$. 
\end{proof}

\subsection{Ordering} 

\begin{lemma}
\label{lemma:ordering}
The vertices of a sliceable polygon $\D$ can be ordered such that for any burn sites $u < v < w$, the region of $P$ burned from $u$ does not share a boundary with the region in $P$ burned from $w$. 
\end{lemma}

\begin{figure}
\centering
\includegraphics[width=0.42\textwidth]{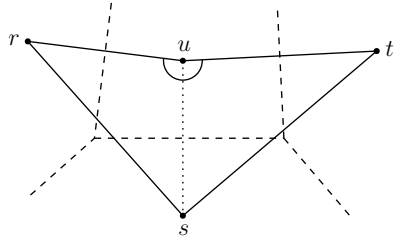}
\caption{Quadrilateral $urst$ is reflex at $u$, a contradiction establishing (P1).}
\label{figure:reflex-contradiction}
\end{figure}

\begin{proof}
We first prove that (P1) each Voronoi region in $\text{VD}_{\D}(V)$ shares a boundary with at most two other Voronoi regions. Then we show that (P2) the graph joining two vertices if they share such a boundary is connected and thus forms a path, which defines an ordering of vertices required by the lemma. (The path can be directed in two ways, either of which defines such an ordering.)

For (P1), suppose for the sake of contradiction that vertex $u$ of $\D$ forms Voronoi edges in $\text{VD}(V)$ that cross $\D$ with three other vertices, say $r$, $s$, and $t$. Since $\D$ is sliceable, the endpoints (Voronoi vertices) of these Voronoi edges lie outside $\D$. 

Let $\D'$ be the convex hull of $\{u,r,s,t\}$.
The Voronoi edge between $u$ and $r$ in $\text{VD}(\{u,r,s,t\})$ contains the corresponding Voronoi edge in $\text{VD}(V)$
since every point that is closest to $u$ and $r$ among all vertices of $V$ is still closest to $u$ and $r$ among a subset of $V$.
The same is true for the Voronoi edges between $u$ and $s$ and between $u$ and $t$.
Thus, since all three of these Voronoi edges cross $\D$ in $\text{VD}(V)$ the corresponding edges in $\text{VD}(\{u,r,s,t\})$ cross $\D$ and hence cross $\D' \subseteq \D$ as well. It follows that a sliceable polygon $\D$ with a vertex $u$ that creates Voronoi edges crossing $\D$ with three different vertices $r$, $s$, and $t$ implies the existence of a sliceable quadrilateral $\D'$ with the same property. To obtain a contradiction and establish (P1), we will argue that no such quadrilateral exists.

Assume $r$, $s$, and $t$ are labelled so that the circumcentres $c_1$ of $\triangle urs$ and $c_2$ of $\triangle ust$ are the two Voronoi vertices shared by these three Voronoi edges.
Since the boundary of the Voronoi region of $u$ intersects $\D'$ in three segments that do not contain $c_1$ or $c_2$, $c_1$ lies on the side of the line through $rs$ opposite $u$ and $c_2$ lies on the side of the line through $st$ opposite of $u$. It follows that $\angle rus$ and $\angle sut$ are obtuse. Thus the interior angle of $P'$ at $u$ is greater than $\pi$, contradicting the convexity of $P'$ (Figure~\ref{figure:reflex-contradiction}).  
This result establishes (P1).

For (P2), assume for a contradiction that the graph has more than one connected component. Then no inter-component vertices form Voronoi boundaries with each other in $\text{VD}_{\D}(V)$. It follows that the fires burning from separate connected components never meet, and hence $P$ cannot be burned entirely.
This contradiction establishes (P2).
\end{proof}

\subsection{Sliceability of Subsets} 

In this section, we study the sliceability of subsets of sliceable polygons. In particular, we show that a sliceable polygon $\D$ contains no Voronoi vertex from $\VD{S}$ for any subset $S \subseteq \V$. While the existence of a dynamic programming algorithm does not require this result, it adds to our characterization of sliceable polygons and allows us to define a simpler recurrence for PB on $\D$ which yields a faster dynamic programming algorithm. 

The \emph{Delaunay triangulation} of a set $S$ of sites, denoted $\DT{S}$, is the dual graph of $\VD{S}$. It is a triangulation of $S$ such that no circumcircle of any triangle in $\DT{S}$ contains a site. The circumcenters of the triangles are the vertices of $\VD{S}$. 

\begin{lemma}
\label{lemma:edge-flip} 
Let $T$ be a triangulation of a convex polygon $\D$. Suppose there exist adjacent triangles $pqr$ and $prs$ in $T$ that form a convex quadrilateral. If $\D$ contains the circumcenter of $\triangle pqr$ and $s$ is interior to the circumcircle of $\triangle pqr$, then $\D$ contains the circumcenter of $\triangle pqs$ or the circumcenter of $\triangle qrs$, or both. 
\end{lemma}
\begin{proof}
Assume the vertices of quadrilateral $pqrs$ are labelled in counter-clockwise order.
 By the conditions of the lemma, triangles $pqs$ and $qrs$ form the Delaunay triangulation of quadrilateral $pqrs$.
Orient $\D$ so that $\overline{pq}$ is aligned with the x-axis with $r$ and $s$ lying above it (Figure~\ref{figure:edge-flip}). Let $f$, $g$, and $h$ denote the circumcenters of $\triangle pqr$, $\triangle pqs$, and $\triangle qrs$, respectively. Since $r$ lies outside $C_{pqs}$ above $\overline{qs}$, $C_{pqs}$ lies below $C_{pqr}$, implying that $g$ is below $f$. Similarly, since $p$ lies outside $C_{qrs}$ left of $\overline{qs}$, $C_{qrs}$ lies right of $C_{pqr}$, which implies that $h$ is right of $f$. To prove that either $g$ or $h$ lies in $\D$ given that $f$ is in $\D$, consider two cases:

\begin{figure}
\centering
\includegraphics[width=0.9\textwidth]{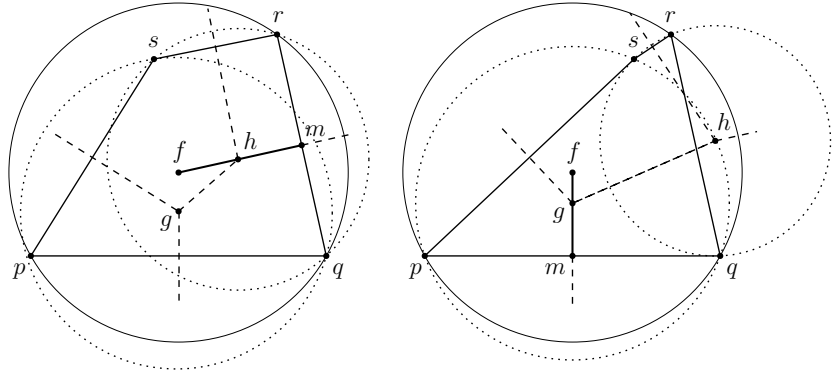}
\caption{Illustration of Case 1 (left) and Case 2 (right) of Lemma~\ref{lemma:edge-flip} with $C_{pqr}$ (solid), $C_{pqs}$ and $C_{qrs}$ (dotted), and $\VD{\{p, q, r, s\}}$ (dashed).} 
\label{figure:edge-flip}
\end{figure}

Case 1: Suppose $h$ lies on or left of $\overline{qr}$. Let $m$ be the midpoint of $\overline{qr}$. Since $h$ is right of $f$ and both $f$ and $h$ lie on the bisector of $q$ and $r$, $h$ lies along $\overline{fm}$. Hence, by the convexity of $\D$, $h$ lies in $\D$. 

Case 2: Otherwise, $h$ lies right of $\overline{qr}$. Then $\angle qsr > \frac{\pi}{2}$. We show that $g$ must lie on or above $\overline{pq}$ in this scenario. Assume for a contradiction that $g$ lies below $\overline{pq}$. Then $\angle psq > \frac{\pi}{2}$. This yields $\angle psr = \angle psq + \angle qsr > \pi$, which implies that $\D$ is not convex. This contradiction establishes that $g$ lies above $\overline{pq}$. By the same analysis provided in the previous case, we conclude that $\D$ contains $g$.
\end{proof}

\begin{lemma}
\label{lemma:legal-triangulation} 
Consider a triangulation $T$ of a convex polygon $\D$. If $\D$ contains the circumcenter of a triangle in $T$, then it contains the circumcenter of a triangle in the Delaunay triangulation $\text{DT}(\V)$ of $\V$. 
\end{lemma}

\begin{proof}
Let $\overline{pr}$ be an edge in $T$ incident to two triangles $pqr$ and $prs$ that form a convex quadrilateral. We say $\overline{pr}$ is an \emph{illegal edge} if $s$ lies in $C_{pqr}$. A new triangulation $T'$ of $\D$ can be obtained from $T$ by replacing $\overline{pr}$ with $\overline{qs}$. This \emph{edge flip} operation creates $\triangle pqs$ and $\triangle qrs$ in place of $\triangle pqr$ and $\triangle prs$. If $\overline{pr}$ is illegal, then, by Lemma~\ref{lemma:edge-flip}, $\D$ contains the circumcenter of $\triangle pqs$ or $\triangle qrs$ (or both) if it contains the circumcenter of $\triangle pqr$ or $\triangle prs$. More generally, assuming that $T'$ is obtained by flipping an illegal edge in $T$, $\D$ contains the circumcenter of some triangle in $T'$ if it contains the circumcenter of some triangle in $T$. We can compute $\text{DT}(\V)$ by flipping illegal edges in $T$ until none exist~\cite{CompGeomText}. Therefore, by repeated application of Lemma~\ref{lemma:edge-flip}, $\D$ contains the circumcenter of some triangle in $\text{DT}(\V)$ if it contains the circumcenter of some triangle in $T$. 
\end{proof}

\begin{theorem}
\label{theorem:sliceable} 
If a convex polygon $\D$ does not contain the circumcenter of any triangle in $\text{DT}(\V)$, then $\D$ does not contain the circumcenter of any triangle in $\text{DT}(S)$ for any $S \subseteq \V$. 
\end{theorem}

\begin{proof}
For completeness, we restate the theorem in terms of Voronoi diagrams.   If a convex polygon $\D$ does not contain any Voronoi vertex of $\text{VD}(\V)$, then $\D$  does not contain any Voronoi vertex of $\VD{S}$ for any $S \subseteq \V$. 

We provide a contrapositive proof. Suppose $\D$ contains the circumcenter of $\triangle pqr$ in $\text{DT}(S)$. Let $T$ be any triangulation of $\D$ containing $\triangle pqr$. Of course, $\D$ contains the circumcenter of a triangle in $T$, implying that $\D$ contains the circumcenter of a triangle in $\text{DT}(\V)$ by Lemma~\ref{lemma:legal-triangulation}. 
\end{proof}

\begin{corollary}
\label{corollary:sliceable}
If $\D$ is sliceable, then the convex hull of $S$ is sliceable for any $S \subseteq \V$.
\end{corollary}

\begin{proof}
Since $\D$ contains no Voronoi vertex of $\VD{S}$ for any $S \subseteq \V$ by Theorem~\ref{theorem:sliceable}, neither does any subset of $\D$, including the convex hull of $S$. 
\end{proof}

\subsection{Dynamic Programming Algorithm} 

Let $\D$ be a sliceable polygon. The ordering of its vertices $v_1, v_2, \dots, v_n$ as defined in Lemma~\ref{lemma:ordering} permits a dynamic programming algorithm similar to the one used for $1$-dimensional polygons that solves PB on $\D$. 

Let $\optS(i,k)$ denote the minimum time to burn the subset of $\D$ from the bisector of $\overline{v_{i-1}v_{i}}$ onward given that $v_i$ is a burn site and $k$ sites remain to be chosen. 
\[
\optS(i,k) =
\begin{cases}
  \min_{i < j \leq n} \max\{d(v_j, p_{ij}), d(v_j, q_{ij}), \optS(j,k-1) \} &
  \text{if $k > 0$,}\\
  d(v_i, v_n) & \text{otherwise,}
\end{cases}
\]
where $p_{ij}$ and $q_{ij}$ represent the intersections of the bisector of $\overline{v_i v_j}$ with $\partial P$. It follows that the minimum time to burn $\D$ using $k$ sites is 
\[
\opt{k}{\D}  =
\begin{cases}
\min_{i \in [n]} \max\{ d(v_1,v_i) , \optS(i,k-1)\} & \text{if $k>0$}\\
\infty & \text{otherwise.}
\end{cases} 
\]

\begin{figure}
\centering
\includegraphics[width=\textwidth]{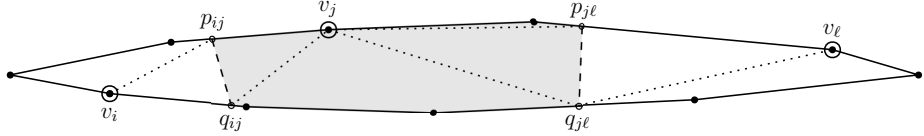}
\caption{Region $P_j$ (shaded) induced by sites $v_i$, $v_j$, and $v_{\ell}$ (circled), overlaid with the distances considered by algorithm (dotted).}
\label{figure:burn-regions-dp}
\end{figure}

\begin{theorem}
\label{theorem:dp-sliceable}
Using a dynamic programming algorithm, PB can be solved in $O(kn^2)$ time on a $n$-vertex sliceable polygon. 
\end{theorem}
\begin{proof} 
(Sketch) We prove that 
the recurrence for $\optS(i,k)$ is correct by
showing that the maximum distance from burn site $v_j$ to a point in the region $P_j$ that is burnt by $v_j$ is correctly calculated in $\optS(i,k)$.
Let $v_i$ be the burn site preceding $v_j$, and $v_\ell$ be the burn site following $v_j$ in the vertex ordering.
The region $P_j$ is bounded by the perpendicular bisectors of segments $\overline{v_iv_j}$ and $\overline{v_jv_\ell}$ which intersect $P$ in segments $\overline{p_{ij}q_{ij}}$ and $\overline{p_{j\ell}q_{j\ell}}$ respectively (Figure~\ref{figure:burn-regions-dp}).
It suffices to show that the time to burn $P_j$ from $v_j$ is the larger of $\max\{d(v_j, p_{ij}), d(v_j, q_{ij})\}$, considered in $\optS(i,k)$, and $\max\{d(v_j, p_{j\ell}), d(v_j, q_{j\ell})\}$, considered in $\optS(j,k-1)$.
If no site precedes $v_j$ then the recurrence correctly uses
$d(v_1,v_j)$ instead of $\max\{d(v_j, p_{ij}), d(v_j, q_{ij}) \}$.
Likewise, if no site follows $v_j$ then the recurrence correctly uses
$d(v_j,v_n)$ instead of $\max\{d(v_j, p_{j\ell}), d(v_j, q_{j\ell}) \}$.

Suppose for the sake of contradiction that there exists a vertex $u$ of $P$ in $P_j$ such that the circle $C$ centered at $v_j$ through $u$ contains $P_j$.

First, $\angle v_j u v_i$ is acute since for $P_j$ to lie inside $C$, both edges of $P_j$ incident to $u$ must form acute angles with its radius $v_j u$.
Hence, since $P$ is convex, both the edge $uv_i$ and the edge $uv_\ell$ form an acute angle with $v_j u$.
Second, both $v_i$ and $v_\ell$ lie outside the circle 
with \emph{diameter} $v_j u$,
otherwise $u$ would be closer to $v_i$ or $v_\ell$ than to $v_j$ and hence not be burned by $v_j$. This implies that $\angle u v_i v_j$ is acute.
Finally, (i) if $v_i < u < v_j$ in the vertex ordering then $\angle v_i v_j u$ is acute, otherwise the perpendicular bisector of $uv_j$ would not separate $v_i$ from $v_j$ which violates the properties of the ordering.
Similarly, (ii) if $v_j < u < v_\ell$ then $\angle v_j v_\ell u$ is acute.

Combining these three observations, we have in case (i) that $\triangle v_i u v_j$ is acute and in case (ii) that $\triangle v_j u v_\ell$ is acute,
both of which contradict Corollary~\ref{corollary:sliceable}.
\end{proof}


\section{Conclusion}
\label{section:conclusion} 

In this paper, we proved PB to be NP-hard on general polygonal domains. Nevertheless, the hardness for simple and convex polygons remains open. In addition, we gave an $O(n^2 \log n + hkn \log n)$-time $3$-approximation algorithm for PB. Finally, we considered sliceable polygons on which we can obtain a dynamic programming solution for PB.
Avenues for future research are to improve the approximation algorithm, to expand the class of polygons solvable using dynamic programming, and to resolve the complexity of PB on simple polygons.


\end{document}